\documentclass{llncs}
\usepackage{graphicx}

\usepackage[utf8]{inputenc}
\usepackage{amssymb}
\usepackage{amsmath}
\usepackage{latexsym}
\usepackage{textcomp}
\usepackage{epsfig}
\usepackage{subfigure}
\usepackage{xspace}
\usepackage{todonotes}
\usepackage{paralist} 
\usepackage{wrapfig}
\usepackage{tabularkv}
\usepackage{hyperref}
\usepackage{marvosym}

\usepackage{thm-restate}

\let\doendproof\endproof
\renewcommand\endproof{~\hfill$\qed$\doendproof}

\begin{document}
	\title{Packing Trees into 1-planar Graphs\thanks{This work started at the Bertinoro Workshop on Graph Drawing 2019 and it is partially supported by: $(i)$ MIUR grant 20174LF3T8, $(ii)$ Dipartimento di Ingegneria - Universit\`a degli Studi di Perugia grants RICBASE2017WD and RICBA18WD, $(iii)$ NFS grants CCF-1740858, CCF-1712119, DMS-1839274, DMS-1839307.}}

\author{
	Felice De Luca \inst{1}\texorpdfstring{ \href{https://orcid.org/0000-0001-5937-7636}{\protect\includegraphics[scale=0.45]{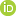}}}{},
	Emilio Di Giacomo \inst{2}\texorpdfstring{ \href{https://orcid.org/0000-0002-9794-1928}{\protect\includegraphics[scale=0.45]{orcid.png}}}{},
	Seok-Hee Hong \inst{3}\texorpdfstring{ \href{https://orcid.org/0000-0003-1698-3868}{\protect\includegraphics[scale=0.45]{orcid.png}}}{},\\ 
	Stephen Kobourov \inst{1}\texorpdfstring{ \href{https://orcid.org/0000-0002-0477-2724}{\protect\includegraphics[scale=0.45]{orcid.png}}}{},
	William Lenhart \inst{4}\texorpdfstring{ \href{https://orcid.org/0000-0002-8618-2444}{\protect\includegraphics[scale=0.45]{orcid.png}}}{},
	Giuseppe Liotta \inst{2}\texorpdfstring{ \href{https://orcid.org/0000-0002-2886-9694}{\protect\includegraphics[scale=0.45]{orcid.png}}}{},\\
	Henk Meijer \inst{5}, 
	Alessandra Tappini \inst{2}\texorpdfstring{ \href{https://orcid.org/0000-0001-9192-2067}{\protect\includegraphics[scale=0.45]{orcid.png}}}{}\href{alessandra.tappini@studenti.unipg.it}{\textsuperscript{(\Letter)}},
	Stephen Wismath \inst{6}\texorpdfstring{ \href{https://orcid.org/0000-0002-9632-3247}{\protect\includegraphics[scale=0.45]{orcid.png}}}{}
}
	
	\date{}
	
	\institute{
        Department of Computer Science, University of Arizona, US\\
        \and
		Dipartimento di Ingegneria, Universit\`a degli Studi di Perugia, Italy\\
		\and
		School of Computer Science, University of Sydney, Australia\\
        \and
        Department of Computer Science, Williams College, US\\
        \and
        Department of Computer Science, University College Roosevelt, The Netherlands\\
        \and
        Department of Computer Science, University of Lethbridge, Canada\\
        }
	
	\maketitle
	

\begin{abstract}
	We introduce and study the \emph{1-planar packing problem}: Given $k$ graphs with $n$ vertices $G_1, \dots, G_k$, find a 1-planar graph that contains the given graphs as edge-disjoint spanning subgraphs. We mainly focus on the case when each $G_i$ is a tree and $k=3$. We prove that a triple consisting of three caterpillars or of two caterpillars and a path may not admit a 1-planar packing, while two paths and a special type of caterpillar always have one. We then study 1-planar packings with few crossings and prove that three paths (resp. cycles) admit a 1-planar packing with at most seven (resp. fourteen) crossings. We finally show that a quadruple consisting of three paths and a perfect matching with $n \geq 12$ vertices admits a 1-planar packing, while such a packing does not exist if $n \leq 10$.
\end{abstract}

\section{Introduction}\label{se:introduction}

In the \emph{graph packing problem} we are given a collection of $n$-vertex graphs $G_1,\dots,$ $G_k$ and we are requested to find a graph $G$ that contains the given graphs as edge-disjoint spanning subgraphs. Various settings of the problem can be defined depending on the type of graphs that have to be packed and on the restrictions put on the packing graph $G$. The most general case is when $G$ is the complete graph on $n$ vertices and there is no restriction on the input graphs. Sauer and Spencer~\cite{DBLP:journals/jct/SauerS78} prove that any two graphs with at most $n-2$ edges can be packed into $K_n$; Wo\'zniak and Wojda~\cite{DBLP:journals/gc/WozniakW93} give sufficient conditions for the existence of a packing of three graphs. 
The setting when $G$ is $K_n$ and each $G_i$ is a tree ($i=1,2,\dots,k$) has been intensively studied. Hedetniemi et al.~\cite{hhs-npttk-81} show that two non-star trees can always be packed into $K_n$. Notice that, the hypothesis that the trees are not stars is necessary for the existence of the packing because each vertex must have degree at least one in each tree, which is not possible if a vertex is adjacent to every other vertex as it is the case for a star. Wang and Sauer~\cite{WANG1993137} give sufficient conditions for the existence of a packing of three trees into $K_n$, while Mah\'eo et al.~\cite{DBLP:journals/ejc/MaheoSW96} characterize the triples of trees that admit~such~a~packing.\\
%
\indent Garc\'ia et al.~\cite{DBLP:journals/jgt/GarciaHHNT02} consider the \emph{planar packing problem}, that is the case when the graph $G$ is required to be planar. They conjecture that the result of Hedetniemi et al. extends to this setting, i.e.,  that every pair of non-star trees can be packed into a planar graph. Notice that, when $G$ is required to be planar, two is the maximum number of trees that can be packed (because three trees have more than $3n-6$ edges). Garc\'ia et al. prove their conjecture for some restricted cases, namely when one of the trees is a path and when the two trees are isomorphic. In a series of subsequent papers the conjecture has been proved true for other pairs of trees. Oda and Ota~\cite{oo-tpptt-06} prove it when one tree is a caterpillar or it is a spider of diameter four. Frati et al.~\cite{DBLP:journals/ipl/FratiGK09} extend the last result to any spider, while Frati~\cite{DBLP:conf/cccg/Frati09} considers the case when both trees have diameter four. Geyer et al. show that a planar packing always exists for a pair of binary trees~\cite{DBLP:conf/wads/GeyerHKKT13} and for a pair of non-star trees~\cite{DBLP:journals/jocg/GeyerHKKT17}, thus finally settling the conjecture.\\
\indent In the present paper we initiate the study of the \emph{1-planar packing problem}, i.e., the problem of packing a set of graphs into a 1-planar graph. A 1-planar graph is a graph that can be drawn so that each edge has at most one crossing. 1-planar graphs have been introduced by Ringel~\cite{Ringel1965107} and have received increasing attention in the last years in the research area called \emph{beyond planarity} (see, e.g.,~\cite{KOBOUROV201749,DBLP:journals/csur/DidimoLM19}). Since any two non-star trees admit a planar packing, a natural question is whether we can pack more than two trees into a 1-planar graph. On the other hand, since each 1-planar graph has at most $4n-8$ edges edges~\cite{DBLP:journals/combinatorica/PachT97}, it is not possible to pack more than three trees into a 1-planar graph. Thus, our main question is whether any three trees with maximum vertex degree $n-3$ admit a 1-planar packing. The restriction to trees of degree at most $n-3$ is necessary because a vertex of degree larger than $n-3$ in one tree cannot have degree at least one in the other two trees. Our results can be listed as follows.~
\begin{itemize}
\item We show that there exist triples of structurally simple trees that do not admit a 1-planar packing (Section~\ref{se:counterexamples}). These triples consist of three caterpillars
with at least $10$ vertices
and of two caterpillars and a path
with $7$ vertices.
\item Motivated by the above results, we study triples consisting of two paths and a caterpillar (Section~\ref{se:2-paths+cater}). We characterize the triples consisting of two paths and a 5-legged caterpillar (a caterpillar where each vertex of the spine has no leaves attached or it has at least five) that admit such a packing. We also characterize the triples that admit a 1-planar packing and that consist of two paths and a caterpillar whose spine has exactly two vertices.
\item The packing technique of the results above is constructive and it gives rise to 1-plane graphs (i.e., 1-planar embedded graphs) with a linear number of crossings. This naturally raises the question about the number of edge crossings required by a 1-planar packing. We show that any three paths with at least six vertices can be packed into a 1-plane graph with seven edge crossings in total (Section~\ref{se:3-paths}). We also extend this technique to three cycles obtaining 1-plane graphs with fourteen crossings in total. 
\item We finally consider the 1-planar packing problem for quadruples of acyclic graphs (Section~\ref{se:quadruples}). Since, as already observed, four paths cannot be packed into a 1-planar graph, we consider three paths and a perfect matching. We show that when $n\geq 12$ such a quadruple admits a 1-planar packing and that when $n \leq 10$ a 1-planar packing does not exist.
\end{itemize}
Preliminary definitions are given in Section~\ref{se:preliminaries} and open problems are listed in Section~\ref{se:open}.
Some proofs are sketched or removed and can be found in the appendix.


\section{Preliminaries}\label{se:preliminaries}

Given a graph $G$ and a vertex $v$ of $G$, $\deg_{G}(v)$ denotes the vertex degree of~$v$~in~$G$.
Let $G_1, \dots, G_k$ be $k$ graphs with $n$ vertices; a \emph{packing} of $G_1, \dots, G_k$ is an $n$-vertex graph $G$ that has $G_1,\dots, G_k$ as edge-disjoint spanning subgraphs. We consider the case when $G$ is a $1$-planar graph; in this case we say that $G$ is a \emph{1-planar packing} of $G_1,\dots, G_k$. If $G_1,\dots, G_k$ admit a (1-planar) packing $G$, we also say that $G_1,\dots, G_k$ \emph{can be packed into $G$}.   
We mainly concentrate~on~the~case when each $G_i$ is a tree ($1\leq i\leq k$). In this case (and generally when each~$G_i$ is connected), we have restrictions on the values of $k$ and $n$ for which a packing~exists.

\begin{restatable}{property}{one}\label{pr:one}
A 1-planar packing of $k$ connected $n$-vertex graphs $G_1,\dots, G_k$ exists only if $k \leq 3$ and $n \geq 2k$. Moreover, $\deg_{G_i}(v) \leq n-k$ for each vertex $v$. 
\end{restatable}
A \emph{caterpillar} $T$ is a tree such that removing all the leaves results in a path called the \emph{spine}. A \emph{backbone} of $T$ is a path $v_0,v_1,v_2,\dots,v_k,v_{k+1}$ of $T$ where $v_1,v_2,\dots,v_k$ is the spine of $T$ and $v_0$ and $v_{k+1}$ are two leaves adjacent in $T$ to $v_1$ and $v_k$, respectively. $T$ is \emph{$h$-legged} if every vertex of its spine has degree either $2$ or at least $h+2$ in $T$.

\section{Trees That Do Not Admit 1-planar Packings}\label{se:counterexamples}

In this section we describe triples of trees that do not admit a 1-planar packing.  

\begin{theorem}\label{th:no-3-trees}
For every $n \geq 10$, there exists a triple of caterpillars that does not admit a 1-planar packing.
\end{theorem}
\begin{proof}
The triple consists of three isomorphic caterpillars $T_1, T_2, T_3$ with $n \geq 10$ vertices. Each $T_i$ has a backbone of length $5$ and $n-5$ leaves all adjacent to the middle vertex of the spine, which we call the \emph{center} of $T_i$. First, notice that each $T_i$ satisfies Property~\ref{pr:one}, i.e., $\deg_{T_i}(v) \leq n-3$. Namely, the vertex with largest degree in $T_i$ is its center, which has degree $n-3$. Let $G$ be any packing of $T_1$, $T_2$, and $T_3$ and let $v_1$, $v_2$, and $v_3$ be the three vertices of $G$ where the three centers of $T_1, T_2, T_3$, respectively, are mapped. The three vertices $v_1$, $v_2$, and $v_3$ must be distinct because otherwise they would have degree larger than $n-1$ in $G$, which is impossible. For each $v_i$ we have $\deg_{T_i}(v_i)=n-3$ and $\deg_{T_j}(v_i) \geq 1$, for $j \neq i$. This implies that $\deg_G(v_i)=n-1$ for each $v_i$. In other words, each $v_i$ is adjacent to all the other vertices of $G$. Thus, $G$ contains $K_{3,n-3}$ as a subgraph. Since $n \geq 10$ and $K_{3,7}$ is not 1-planar~\cite{CZAP2012505}, $G$ is not $1$-planar.
\end{proof}

Motivated by Theorem~\ref{th:no-3-trees}, we consider triples where one of the caterpillars is a path. Also in this case there exist triples that do not have a 1-planar packing.
\begin{theorem}\label{th:no-path-plus-2-caterpillars}
There exists a triple consisting of a path and two caterpillars with $n=7$ vertices that does not admit a 1-planar packing. 
\end{theorem}
\begin{proof}
Let $T_i$ ($i=1,2$) be a caterpillar with a backbone of length four such that one of the two internal vertices has degree three and the other one has degree four. Let $G$ be a packing of $T_1$, $T_2$ and a path $P$ of $7$ vertices. Let $v_1$, $v_2$, $v_3$, and $v_4$ be the four vertices of $G$ where the internal vertices of the backbones of $T_1$ and $T_2$ are mapped to. We first observe that $v_1$, $v_2$, $v_3$, and $v_4$ must be distinct. Suppose, as a contradiction, that two of them coincide, say $v_1$ and $v_2$; then $\deg_{T_1}(v_1)+\deg_{T_2}(v_1)\geq 6$. On the other hand $\deg_{P}(v_1) \geq 1$, and therefore $\deg_G(v_1) \geq 7$, which is impossible (since $G$ has only $7$ vertices).
Denote by $G_{1,2}$ the subgraph of $G$ containing only the edges of $T_1$ and $T_2$. Two vertices among $v_1$, $v_2$, $v_3$, and $v_4$, say $v_1$ and $v_2$, have degree $5$ in $G_{1,2}$, while the other two have degree $4$ in $G_{1,2}$.  Consider now the edges of $P$. Since the maximum vertex degree in a graph of seven vertices is six, $v_1$ and $v_2$ must be the end-vertices of $P$, while $v_3$ and $v_4$ are internal vertices. This means that they all have degree $6$ in $G$. The vertices distinct from $v_1$, $v_2$, $v_3$, and $v_4$ have degree $2$ in $G_{1,2}$ and degree $4$ in $G$. Thus in $G$ there are four vertices of degree $6$ and three vertices of degree $4$. The only graph of seven vertices with this degree distribution is the graph obtained from~$K_7$ by deleting all the edges of a $3$-cycle, which is known to be non-1-planar~\cite{KORZHIK20081319}.
\end{proof}

\section{1-planar Packings of Two Paths and a Caterpillar }\label{se:2-paths+cater}

In this section we prove that a triple consisting of two paths $P_1$ and $P_2$ and a 5-legged caterpillar $T$ with at least six vertices admits a 1-planar packing. Let $P$ be the backbone of $T$ and let $P'_1$ and $P'_2$ be two paths with the same length as $P$. We first show how to construct a 1-planar packing of $P$, $P'_1$ and $P'_2$. We then modify the computed packing to include the leaves of the caterpillar; this requires transforming some edges of $P'_1$ and $P'_2$ to sub-paths that pass through the added leaves. The resulting packing is a 1-planar packing of $P_1$, $P_2$ and $T$.\\
\indent Let $\Gamma$ be a $1$-planar drawing, possibly with parallel edges, and let $e$ be an edge of $\Gamma$. If $e$ has one crossing $c$, then each of the two parts in which $e$ is divided by $c$ are called \emph{sub-edges} of $e$; if $e$ has no crossing, $e$ itself is called a \emph{sub-edge} of $e$. Let $v$ be a vertex of $\Gamma$; a \emph{cutting curve} of $v$ is a Jordan arc $\gamma$ such that: (i) $\gamma$ has $v$ as an end-point; (ii)  $\gamma$ intersects two edges $e_1=(u_1,v_1)$ and $e_2=(u_2,v_2)$ (possibly $u_1=u_2$ and/or $v_1=v_2$);  (iii) $\gamma$ does not intersect any other edge of $\Gamma$; (iv) $e_1$ and $e_2$ do not cross each other; (v) if $e_1$ and $e_2$ are parallel edges (i.e., $u_1=u_2$ and $v_1=v_2$), they have no crossings. The \emph{stub} of $e_i$ with respect to $\gamma$ is the sub-edge of $e_i$ intersected by $\gamma$ ($i=1,2$). Given a cutting curve $\gamma$ of a vertex $v$, and an integer $k \geq 5$, a \emph{$k$-leaf addition} operation adds $k$ vertices $w_1, w_2, \dots, w_k$ and the edges $(v,w_1), (v,w_2), \dots, (v,w_k)$ to $\Gamma$ in such a way that: (i) the added vertices subdivide the stubs of both $e_1$ and $e_2$ with respect to $\gamma$; (ii) the subgraph induced by $u_1, u_2, v_1,v_2,w_1,w_2,\dots,w_k$ has no multiple edges (see Fig.~\ref{fi:leaf-addition} for an example). In other words, a leaf addition adds a set of vertices adjacent to $v$ and replaces the stubs of $e_1$ and $e_2$ with two edge-disjoint paths. This operation will be used to modify the 1-planar packing of $P$, $P'_1$ and $P'_2$ to include the leaves of the caterpillar. When the value of $k$ is not relevant, a $k$-leaf addition will be simply called a \emph{leaf addition}.
\begin{figure}[t]
	\centering
	\subfigure[]{\includegraphics[width=0.28\textwidth, page=1]{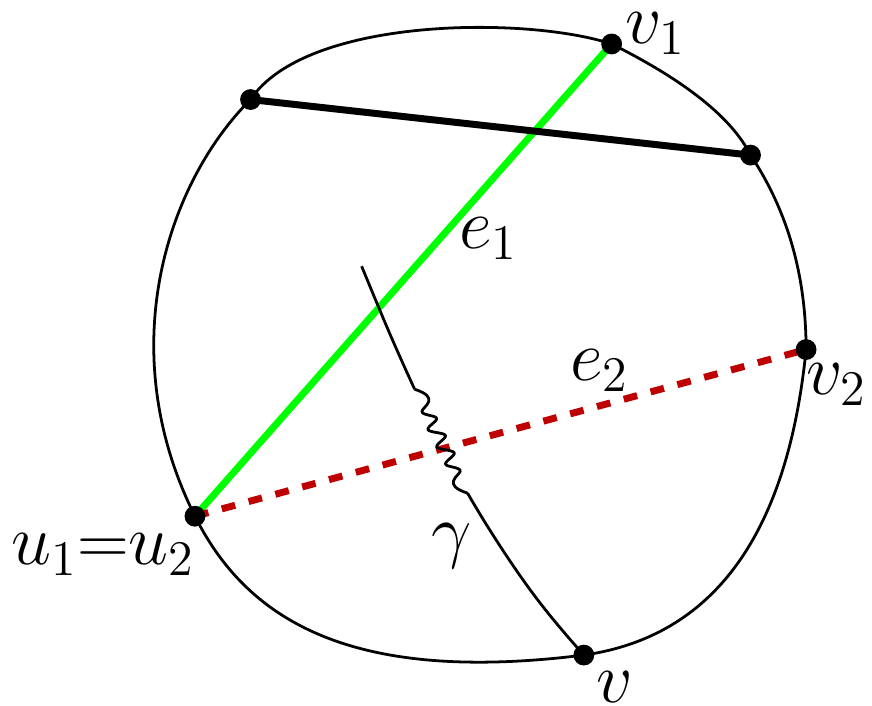}\label{fi:leaf-addition-a}}
	\hfil
	\subfigure[]{\includegraphics[width=0.28\textwidth, page=2]{leaf-addition.pdf}\label{fi:leaf-addition-b}}
	\hfil
	\caption{A $5$-leaf addition operation. The cutting curve is shown with a zig-zag pattern on it.}\label{fi:leaf-addition}
\end{figure}
\begin{lemma}\label{le:leaf-addition}
	Let $\Gamma$ be a 1-planar drawing possibly with parallel edges, let $v$ be a vertex of $\Gamma$ and let $\gamma$ be a cutting curve of $v$. It is possible to execute a $k$-leaf addition for every $k \geq 5$ in such a way that the resulting drawing is still 1-planar. 
\end{lemma}
\begin{proof}
	Denote by $e_1$ and $e_2$ the two edges crossed by $\gamma$. If one of them or both are crossed in $\Gamma$ replace their crossing points with dummy vertices. Let $e'_i$ be the stub of $e_i$ with respect to $\gamma$ (if $e_i$ is not crossed in $\Gamma$, $e'_i$ coincides with $e_i$). After the replacement of the crossings with the dummy vertices the two stubs $e'_1$ and $e'_2$ have no crossing. Since $\gamma$ does not cross any edge distinct from $e_1$ and $e_2$, the drawing $\Gamma'$ obtained by removing $e'_1$ and $e'_2$ has a face $f$ whose boundary contains the vertex $v$ and all the end-vertices of $e'_1$ and of $e'_2$ (there are at least two and at most four such vertices). The idea now is to insert into the face $f$, without creating any crossing, a gadget that realizes the $k$-leaf addition for the desired value of $k \geq 5$. A gadget has $k$ vertices that will be added to $\Gamma$, a vertex that will be
	identified
	with $v$, and four vertices $a$, $b$, $c$, and $d$ that will be
	identified
	with the end-vertices of $e'_1$ and $e'_2$. The four vertices $a$, $b$, $c$, and $d$ will be called \emph{attaching vertices} and the edges incident to them will be called \emph{attaching edges}. In order to guarantee that the leaf addition is valid and that the drawing $\Gamma''$ obtained by the insertion of the gadget inside $f$ is $1$-planar, we have to pay attention to two aspects: (i) if an attaching edge is crossed in the gadget, then its attaching vertex cannot be
	identified
	with a dummy vertex (otherwise when we remove the dummy vertex we obtain an edge that is crossed twice); (ii) if two attaching vertices of the gadget are coincident (because two end-vertices of $e'_1$ and $e'_2$ coincide), then the corresponding attaching edges must not have the second end-vertex in common in the gadget (otherwise the leaf addition is not valid because it creates multiple edges). We use different gadgets depending on whether $e_1$ and $e_2$ are parallel edges or not. If they are parallel edges, we use the gadgets of Figs.~\ref{fi:gadgets-parallel-1}--\ref{fi:gadgets-parallel-4} and~\ref{fi:gadgets-parallel-5}. Notice that in this case, $e_1$ and $e_2$ are not crossed by definition of cutting curve. It follows that $f$ has no dummy vertex and (i) is guaranteed. On the other hand, both end-vertices of $e_1$ and $e_2$ coincide and therefore the end-vertices of the attaching edges that are not attaching vertices must be distinct. This is true for the gadgets used in this case. If $e_1$ and $e_2$ are non-parallel, we use the gadgets of Figs.~\ref{fi:gadgets-non-parallel-1}--\ref{fi:gadgets-non-parallel-4} and~\ref{fi:gadgets-non-parallel-5}. All these gadgets have only one attaching edge that is crossed (labeled $d$ in the figure); also, vertex $d$ can be
	identified
	with vertex $c$ without creating multiple edges. If $e_1$ and $e_2$ are non-parallel, at most two end-vertices of $e'_1$ and $e'_2$ are dummy; they cannot belong to the same stub, and they cannot coincide (because $e_1$ and $e_2$ do not cross each other). Thus we can
	identify $d$
	with a non-dummy vertex and we can
	identify $c$ and $d$
	if needed.
\end{proof}

\begin{figure}[tbp]
	\centering
	\subfigure[$k=5$]{\includegraphics[width=0.23\textwidth, page=1]{gadgets-parallel.pdf}\label{fi:gadgets-parallel-1}}
	\hfil
	\subfigure[$k=6$]{\includegraphics[width=0.23\textwidth, page=2]{gadgets-parallel.pdf}\label{fi:gadgets-parallel-2}}
	\hfil
	\subfigure[$k=7$]{\includegraphics[width=0.23\textwidth, page=3]{gadgets-parallel.pdf}\label{fi:gadgets-parallel-3}}
	\hfil
	\subfigure[$k>6$ even]{\includegraphics[width=0.23\textwidth, page=4]{gadgets-parallel.pdf}\label{fi:gadgets-parallel-4}}
	\hfil
	\subfigure[$k=5$]{\includegraphics[width=0.23\textwidth, page=1]{gadgets-non-parallel.pdf}\label{fi:gadgets-non-parallel-1}}
	\hfil
	\subfigure[$k=6$]{\includegraphics[width=0.23\textwidth, page=2]{gadgets-non-parallel.pdf}\label{fi:gadgets-non-parallel-2}}
	\hfil
	\subfigure[$k=7$]{\includegraphics[width=0.23\textwidth, page=3]{gadgets-non-parallel.pdf}\label{fi:gadgets-non-parallel-3}}
	\hfil
	\subfigure[$k>6$ even]{\includegraphics[width=0.23\textwidth, page=4]{gadgets-non-parallel.pdf}\label{fi:gadgets-non-parallel-4}}
	\hfil
	\subfigure[$k>7$ odd]{\includegraphics[width=0.46\textwidth, page=5]{gadgets-parallel.pdf}\label{fi:gadgets-parallel-5}}
	\hfil
	\subfigure[$k>7$ odd]{\includegraphics[width=0.46\textwidth, page=5]{gadgets-non-parallel.pdf}\label{fi:gadgets-non-parallel-5}}
	\caption{Gadgets for the proof of Lemma~\ref{le:leaf-addition}. \subref{fi:gadgets-parallel-1}--\subref{fi:gadgets-parallel-4} and \subref{fi:gadgets-parallel-5} are used for parallel edges; \subref{fi:gadgets-non-parallel-1}--\subref{fi:gadgets-non-parallel-4} and \subref{fi:gadgets-non-parallel-5} are used for non-parallel edges.}\label{fi:gadgets-parallel}
\end{figure}	
%
%
We are ready to describe our construction of a 1-planar packing of $P_1$, $P_2$, and $T$. We use different techniques for different lengths of the backbone of~$T$. 
%
%
\begin{figure}[tbp]
	\centering
	\subfigure[$n'=4k$]{\includegraphics[width=0.23\textwidth, page=2]{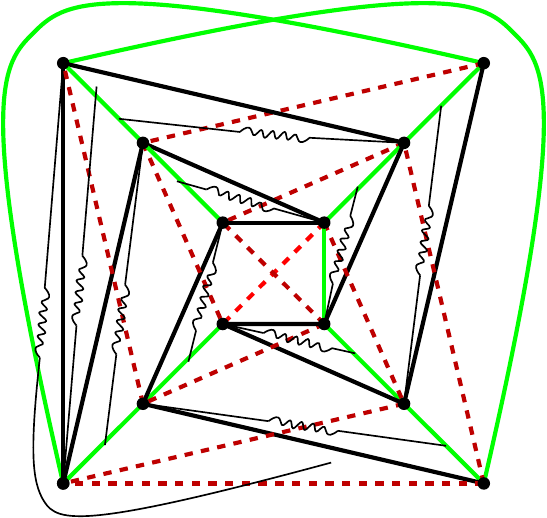}\label{fi:5-legged-1-b}}
	\hfill
	\subfigure[$n'=4k+1$]{\includegraphics[width=0.23\textwidth, page=4]{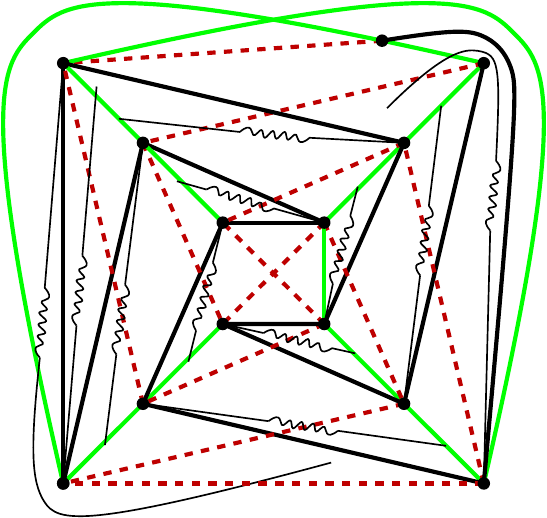}\label{fi:5-legged-1-c}}
	\hfill
	\subfigure[$n'=4k+2$]{\includegraphics[width=0.23\textwidth, page=6]{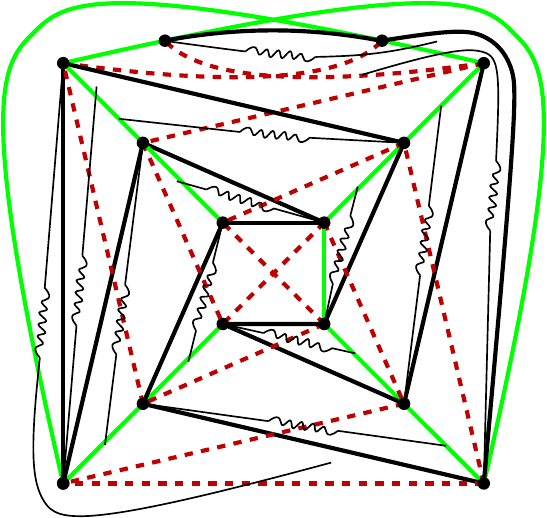}\label{fi:5-legged-1-d}}
	\hfill
	\subfigure[$n'=4k+3$]{\includegraphics[width=0.23\textwidth, page=8]{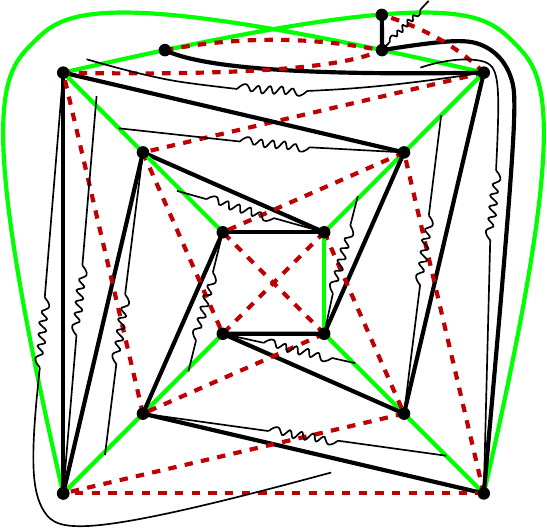}\label{fi:5-legged-1-e}}
	\caption{1-planar packings of three paths with $n'$$\geq$$8$ vertices (case $k$=$3$); A cutting curve is shown (zig-zag pattern) for each internal vertex of the black path.
	}\label{fi:5-legged-1-ce}
\end{figure}

\begin{lemma}\label{le:5legged-6ormore}
Two paths and a 5-legged caterpillar whose backbone contains $n' \geq 6$ vertices admit a 1-planar packing.
\end{lemma}
\begin{proof}
We start with the construction of a 1-planar packing of the three paths $P'_1$, $P'_2$ and $P$. Let $n'$ be the number of vertices of $P'_1$, $P'_2$ and $P$, assume first that $n' \geq 8$ and $n' \equiv 0 \pmod 4$. A 1-planar packing of $P'_1$, $P'_2$ and $P$ for this case is shown in Fig.~\ref{fi:5-legged-1-b} for $n'=16$ and it is easy to see that it can be extended to any $n'$ multiple of 4.
Assume that the backbone $P$ of $T$ is the path shown in black in Fig.~\ref{fi:5-legged-1-b}. To add the leaves of $T$ to the construction we define a cutting curve for each vertex $v$ that has some leaves attached; we then execute a leaf addition operation for each such vertex. By Lemma~\ref{le:leaf-addition}, it is possible to execute each leaf addition so to guarantee the 1-planarity of the resulting drawing. The cutting curve for each internal vertex of $P$ is shown in Fig.~\ref{fi:5-legged-1-b} with a zig-zag pattern.
Note that, regardless of the order in which the leaf additions are executed, the cutting curves remain valid.\\
\indent Suppose now that $n' \geq 8$ and $n' \not \equiv 0 \pmod 4$. In this case we first construct a 1-planar packing of three paths with $n''=4k$ vertices (with $k = \lfloor \frac{n'}{4} \rfloor$) using the same construction as in the previous case and then we add one, two or three vertices as shown in Figs.~\ref{fi:5-legged-1-c}-\ref{fi:5-legged-1-e},
where we also show the cutting curves for each internal vertex of $P$. If $n'$ is equal to $6$ or $7$, we use the same approach; the only difference is in the construction of the 1-planar packing of $P'_1$, $P'_2$ and $P$. The construction for such a packing and the cutting curves for the internal vertices of $P$ are shown in Figs.~\ref{fi:5-legged-2-a} and~\ref{fi:5-legged-2-b}.
\end{proof}
\begin{restatable}{lemma}{fivelegged}\label{le:5legged-5}
Two paths and a 5-legged caterpillar $T$ whose backbone contains $n' = 5$ vertices admit a 1-planar packing, unless $T$ is a path.
\end{restatable}
\begin{proof}
	If $T$ is a path, then $P_1$, $P_2$ and $T$ are all paths of length five, and by Property~\ref{pr:one}, a 1-planar packing of $P_1$, $P_2$ and $T$ does not exist. Suppose therefore that at least one internal vertex of the backbone $P$ of $T$ has some leaves attached. We use an approach similar to the one of Lemma~\ref{le:5legged-6ormore}. However, as just explained, a 1-planar packing of $P'_1$, $P'_2$ and $P$ does not exist in this case. We start with a 1-planar packing with two pairs of parallel edges. For each pair, one edge belongs to $P'_1$ and the other one to $P'_2$. We will remove the parallel edges by performing the leaf addition operations. To this aim we must guarantee that there is a cutting curve for each pair of parallel edges. The 1-planar packing $P'_1$, $P'_2$ and $P$ and the cutting curves for the internal vertices of $P$ are shown in Fig.~\ref{fi:5-legged-2-c}, for the case when at least two vertices have leaves attached. Indeed, if only two vertices have leaves attached, they are either consecutive along the backbone or not. In the first case, these two vertices are mapped to the vertices labeled $a$ and $b$ in Fig.~\ref{fi:5-legged-2-c} and the depicted cutting curves will remove the parallel edges; in the second case, the two vertices are mapped to the vertices labeled $a$ and $c$ and also in this case the depicted cutting curves will remove the parallel edges.
	\begin{figure}[tb]
		\centering
		\subfigure[$n'=7$]{\includegraphics[width=0.23\textwidth, page=2]{5-legged-2.pdf}\label{fi:5-legged-2-a}}
		\hfill
		\subfigure[$n'=6$]{\includegraphics[width=0.23\textwidth, page=1]{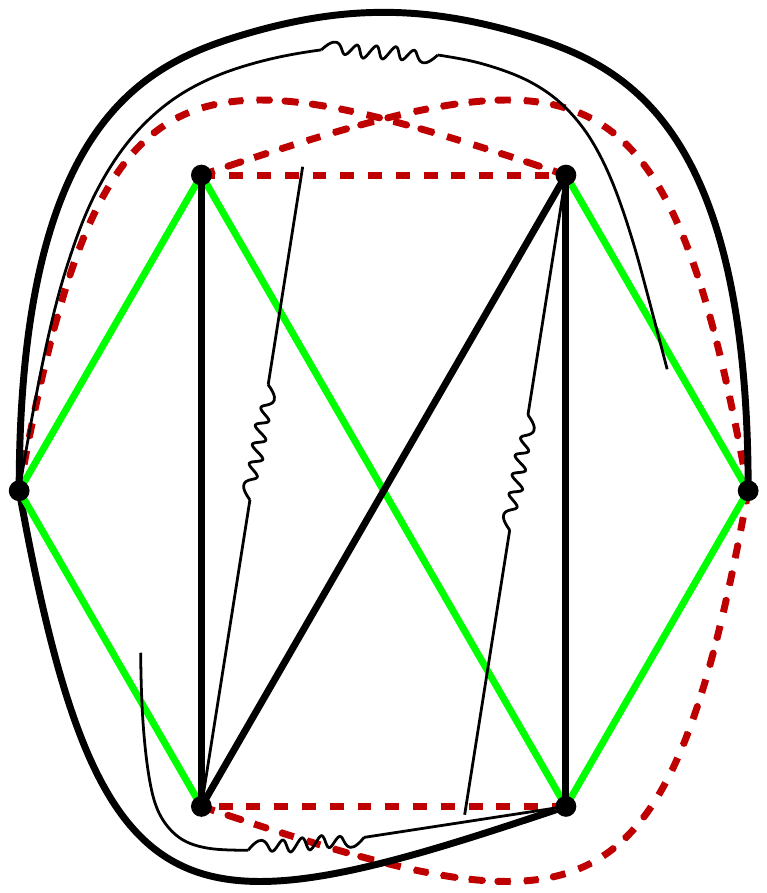}\label{fi:5-legged-2-b}}
		\hfill
		\subfigure[$n'=5$]{\includegraphics[width=0.23\textwidth, page=4]{5-legged-2.pdf}\label{fi:5-legged-2-c}}
		\caption{1-planar packings of three paths with $n' \in \{5,6,7\}$ vertices, with a cutting curve (zig-zag pattern) for each internal vertex of the black path.
		}\label{fi:5-legged-2}
	\end{figure}
	 If only one vertex of $P$ has leaves attached, we have only one cutting curve and thus it is not possible to intersect both pairs of parallel edges. To handle this case we use an ad-hoc technique which can be found in the appendix.
\end{proof}

The next theorem gives a complete characterization for the case in which the backbone of $T$ has length four.
\begin{restatable}{theorem}{twopathscaterpillarfourvertices}\label{th:caterpillar-4}
	Two paths and a caterpillar $T$ whose backbone contains $n' = 4$ vertices admit a 1-planar packing if and only if $n \geq 6$ and $\deg_T(v) \leq n-3$ for every vertex $v$.
\end{restatable}

Lemmas~\ref{le:5legged-6ormore} and~\ref{le:5legged-5}, together with Theorem~\ref{th:caterpillar-4} imply the next theorem.

\begin{theorem}\label{th:2paths-plus-caterpillar}
	Two paths and a 5-legged caterpillar $T$ with $n$ vertices admit a 1-planar packing if and only if $n \geq 6$ and $\deg_T(v) \leq n-3$ for every vertex $v$.
\end{theorem}

\section{1-planar Packings with Constant Edge Crossings}\label{se:3-paths}

The technique described in the previous section constructs 1-planar drawings that have a linear number of crossings. A natural question is whether it is possible to compute a 1-planar packing with a constant number of crossings. In this section we prove that seven (resp. fourteen) crossings suffice for packing three paths (resp. cycles). It is worth remarking that a 1-planar packing of three paths has at least three crossings (because it has $3n-3$ edges), while a 1-planar packing of three cycles has at least six crossings (because it has $3n$ edges).

\begin{theorem}\label{th:3-paths}
  Three paths with $n\geq 6$ vertices can be packed into a 1-plane graph with at most $7$ edge crossings. 
\end{theorem}
\begin{proof}
We prove the statement by showing how to construct a 1-planar drawing  with at most $7$ crossings of a graph that is the union of three paths. Suppose first that $n=7+3k$ for $k \in \mathbb{N}$. If $k=0$, we draw the union of the three paths with $7$ vertices as shown in Fig.~\ref{fi:5-legged-2-a}. The drawing is 1-planar and has three crossings in total. Suppose now that $k > 0$. We consider three rays $r_0,r_1,r_2$ with a common origin pairwise forming a $120^\circ$ angle and we place $k$ vertices on each line. We denote by $u_{i,1},u_{i,2},\dots,u_{i,k}$ the vertices of line $r_i$ ($i=0,1,2$) in the order they appear along $r_i$ starting from the origin (see Fig.~\ref{fi:paths-few-crossings-a}). In the following, indices will be taken modulo 3 when working with the indices of the rays $r_i$. To draw path $P_i$ ($i=0,1,2$) we draw the edges $(u_{i,1},u_{i+1,1})$, $(u_{i,j},u_{i+1,j-1})$, and $(u_{i,j},u_{i+1,j})$ (for $j=2,\dots,k$) as straight-line segments. Notice that, these edges form a zig-zagging path between the vertices of rays $r_i$ and $r_{i+1}$, so $P_i$ passes through all vertices of $r_i$ and $r_{i+1}$ but not through the vertices of $r_{i+2}$. To include these missing vertices in $P_i$, we add to $P_i$ edges $(u_{i+2,j},u_{i+2,j+1})$ (for $j=1,2,\dots,k-1$).
%
\begin{figure}[tbp]
	\centering
	\subfigure[]{\includegraphics[width=0.32\textwidth, page=1]{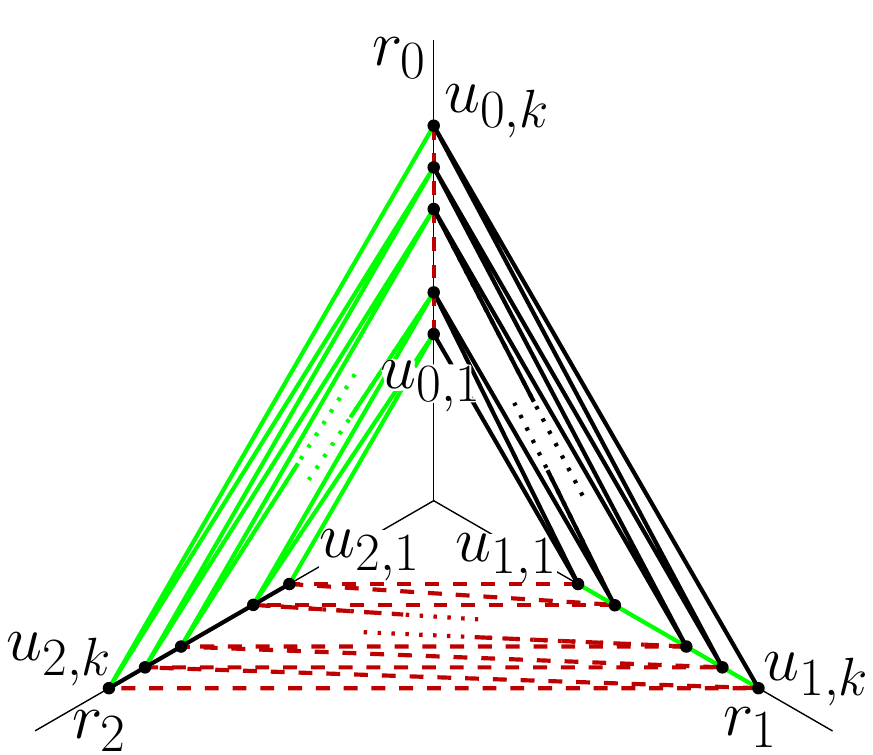}\label{fi:paths-few-crossings-a}}
	\hfil
	\subfigure[]{\includegraphics[width=0.32\textwidth, page=2]{paths-few-crossings.pdf}\label{fi:paths-few-crossings-b}}
	\hfil
	\subfigure[]{\includegraphics[width=0.32\textwidth, page=4]{paths-few-crossings.pdf}\label{fi:paths-few-crossings-c}}
	\caption{Illustration for the proof of Theorem~\ref{th:3-paths}.}\label{fi:paths-few-crossings}
\end{figure}
In this way we draw two disjoint sub-paths for each path $P_i$, namely a zig-zagging path between $r_i$ and $r_{i+1}$ and a straight-line path along $r_{i+2}$. Moreover, we only draw $3k$ edges and therefore there are still $7$ missing vertices (and $8$ missing edges) in each path. To add the missing vertices and edges and to connect the two sub-paths of each path, we construct a drawing $\Gamma_0$ of three paths $P'_0,P'_1,P'_2$ with seven vertices as in the case when $k=0$. Denote with $v_i$ and $w_i$ the end-vertices of $P'_i$ in $\Gamma_0$. We place $\Gamma_0$ inside the triangle $u_{0,1}, u_{1,1}, u_{2,1}$ and add the edges $(v_i,u_{i,1})$ and $(w_i,u_{i+2,1})$. It is easy to see (see also Fig.~\ref{fi:paths-few-crossings-b}) that these six edges can be added so that the drawing is still 1-planar and so that the total number of crossings is $6$. This concludes the proof for $n=7+3k$.
If $n=7+3k+1$ we start with the same construction as in the previous case and then add an extra vertex $v$ outside the triangle $u_{1,k}, u_{2,k}, u_{3,k}$. Notice that each of these three vertices is the end-vertex of two of the three paths with $7+3k$ vertices. Thus we can extend each path to include $v$ by connecting it to each of the three vertices $u_{1,k}, u_{2,k}, u_{3,k}$ in a planar way (see Fig.~\ref{fi:paths-few-crossings-c} ignoring vertex $w$). If $n=7+3k+2$, then we add two extra vertices outside the triangle $u_{0,k}, u_{1,k}, u_{2,k}$ and connect both of them to the three vertices $u_{0,k}, u_{1,k}, u_{2,k}$ (recall that each of these three vertices is the end-vertex of two distinct paths with $7+3k$ vertices). In this case however the addition of the two extra vertices causes the creation of one crossing. Thus the final drawing is 1-planar and the total number of crossings is at most $7$ (see Fig.~\ref{fi:paths-few-crossings-c}). This concludes the proof for $n \geq 7$. If $n=6$ we construct a  1-planar packing of three paths with three crossings in total as shown in Fig.~\ref{fi:5-legged-2-b}.
\end{proof}

The construction of Theorem~\ref{th:3-paths} can be extended to three cycles.
\begin{restatable}{theorem}{threecycles}
	\label{th:3-cycles}
	Three cycles with $n\geq 20$ vertices can be packed into a 1-plane graph with at most $14$ edge crossings.
\end{restatable}
\section{From Triples to Quadruples}\label{se:quadruples}
In this section we extend the study of 1-planar packings from triples of graphs to quadruples of graphs. By Property~\ref{pr:one}, a 1-planar packing of four graphs does not exist if all graphs are connected, because the number of edges of the four graphs is higher than the number of edges allowed in a 1-planar graph. We consider therefore a quadruple consisting of three paths and a perfect matching. Notice that, in this case the number of vertices $n$ has to be even.
\begin{theorem}\label{th:3-paths-plus-matching}
Three paths and a perfect matching with $n \geq 12$ vertices admit a 1-planar packing. If $n \leq 10$, the quadruple does not admit a 1-planar packing. 
\end{theorem}
\begin{proof}
	\begin{figure}[htb]
		\centering
		\subfigure[]{\includegraphics[width=0.32\textwidth, page=1]{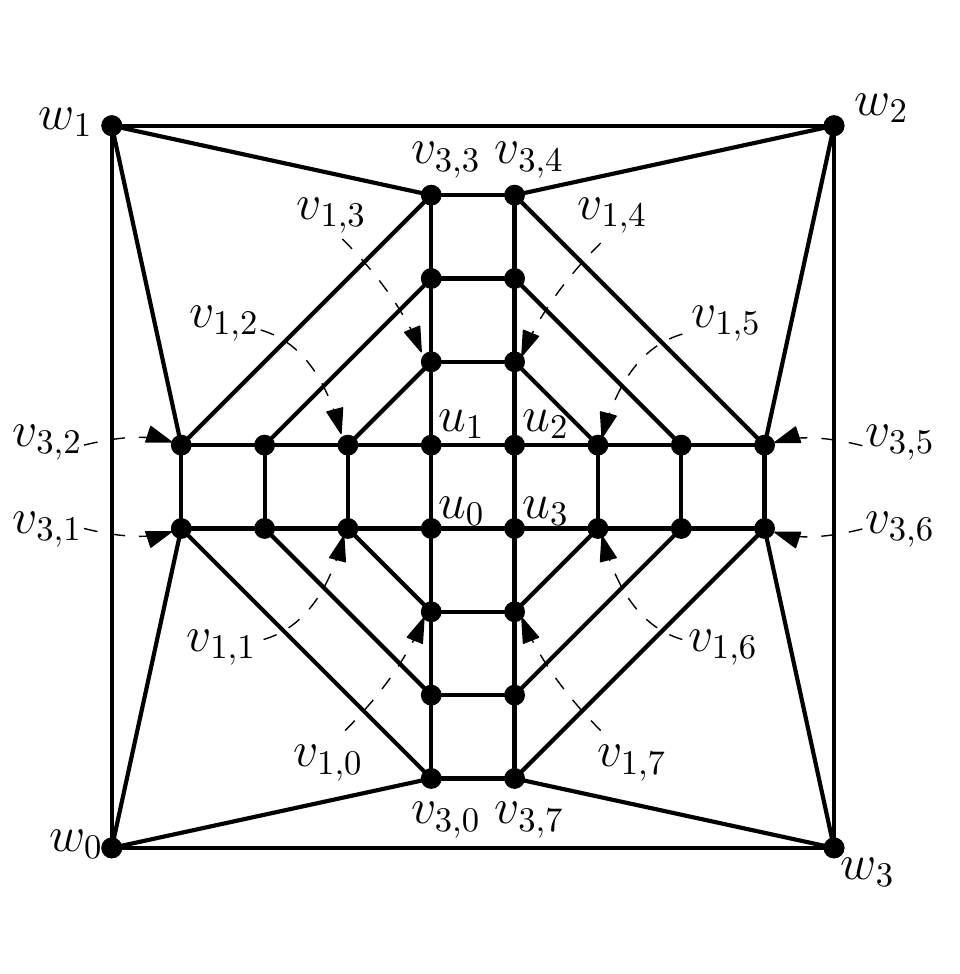}\label{fi:matching2-a}}
		\hfil
		\subfigure[$n=8k$]{\includegraphics[width=0.32\textwidth, page=2]{matching2.pdf}\label{fi:matching2-b}}
		\hfil
		\subfigure[$n=8k+2$]{\includegraphics[width=0.32\textwidth, page=3]{matching2.pdf}\label{fi:matching2-c}}
		\hfil
		\subfigure[$n=8k+4$]{\includegraphics[width=0.32\textwidth, page=4]{matching2.pdf}\label{fi:matching2-d}}
		\hfil
		\subfigure[$n=8k+6$]{\includegraphics[width=0.32\textwidth, page=5]{matching2.pdf}\label{fi:matching2-e}}
		\caption{(a) Graph $G'$ used in the proof of Theorem~\ref{th:3-paths-plus-matching} $(n=8k, k=3)$. (b)--(e) 1-planar packings of three paths and a perfect matching obtained starting~from~$G'$.}\label{fi:matching2}
	\end{figure}
Three paths and a perfect matching have a total of $3(n-1)+\frac{n}{2}=\frac{7n}{2}-3$ edges. Since a 1-planar graph has at most $4n-8$ edges, a 1-planar packing of three paths and a perfect matching exists only if $\frac{7n}{2}-3 \leq 4n-8$, i.e., if $n \geq 10$. If $n=10$, we have $\frac{7n}{2}-3=32$ and $4n-8=32$, which means that any 1-planar packing of three paths and a perfect matching with $n=10$ vertices is an optimal 1-planar graph. It is known that every optimal 1-planar graph has at least eight vertices of degree exactly six~\cite{DBLP:journals/algorithmica/Brandenburg18}. On the other hand, in any 1-planar packing of three paths and a perfect matching all vertices, except the at most six end-vertices of the three paths, have degree seven, which implies that a 1-planar packing of three paths and a perfect matching does not exist.\\
\indent We now prove that a 1-planar packing exists if $n \geq 12$. We only discuss here the case when $n \geq 24$; the cases in which $12 \le n \le 22$ are described in the appendix.
Based on the fact that in any 1-planar packing of three paths and a perfect matching at least $n - 6$ vertices have degree seven, we
construct the desired 1-planar packing starting from a 1-planar graph $G$ such that at least $n - 6$ vertices have degree at least seven; we then partition the edges of $G$ into five sets; three of these sets form a spanning path each, the fourth one forms a perfect matching, and the fifth one contains edges that will not be part of the 1-planar packing. For every $n = 8k$ and $k\geq 3$ it is possible to construct a 1-planar graph with $n$ vertices each having degree at least seven as follows. We start with $k-1$ cycles $C_1,C_2,\dots,C_{k-1}$. Each 
cycle $C_i$ ($1 \leq i \leq k-1$) has eight vertices $v_{i,j}$ with $0 \leq j \leq 7$. Cycle $C_i$, for $1 \leq i \leq k-2$, is embedded inside cycle $C_{i+1}$ and is connected to it with edges $(v_{i,j},v_{i+1,j})$ for each $0 \leq j \leq 7$. We have a cycle with four vertices $u_0, u_1, u_2, u_3$ embedded inside $C_1$ and connected to it with edges $(u_j,v_{1,2j})$ and $(u_j,v_{1,2j+1})$. Finally, we have a cycle with four vertices $w_0, w_1, w_2, w_3$ embedded outside $C_{k-1}$ and connected to it with edges $(w_j,v_{k-1,2j})$ and $(w_j,v_{k-1,2j+1})$. The graph $G'$ described so far has $n$ vertices, is planar, all its vertices have degree four, and each vertex is incident to at most one face of size three (see Fig.~\ref{fi:matching2-a}). By adding two crossing edges inside each face of size four, we obtain a 1-planar graph $G$ with $n$ vertices where each vertex has degree at least seven. The graph $G$ and the partition of the edges of $G$ in five sets defining three paths and a matching is shown in Fig.~\ref{fi:matching2-b}. If $n$ is not a multiple of $8$, then it will be $n=8k+r$, with $0<r<8$ and $r$ even (because $n$ is even). In this case we construct $G'$ as explained above and then we extend the paths $u_0, v_{1,1}, \dots,v_{k-1,1}$ and $u_1, v_{1,2}, \dots,v_{k-1,2}$ to the left with $1$, $2$ or $3$ vertices each; we then suitably rearrange the edges of $G'$.
The graph $G$ is then obtained, as in the previous case, by adding a pair of crossing edges inside each face of size four. The resulting graph $G$ and a partition of its edges in five sets defining three paths and a matching is shown in Figs.~\ref{fi:matching2-c}, \ref{fi:matching2-d}, and~\ref{fi:matching2-e}, for the cases when $r=2$, $r=4$, and $r=6$, respectively. 
\end{proof}

\section{Open Problems}\label{se:open}
We find that the 1-planar packing problem is a fertile and still largely unexplored research subject. We conclude the paper with a list of open problems.
\begin{inparaenum}[(i)]
\item Theorem~\ref{th:no-path-plus-2-caterpillars} holds only for $n=7$. Do two caterpillars (or more general trees) and a path admit a 1-planar packing if they have more than $7$ vertices?  
%
\item Can Theorem~\ref{th:2paths-plus-caterpillar} be extended to general caterpillars? What about two paths and a tree more complex than a caterpillar, for example a binary tree?
%
\item Is it possible to compute a 1-planar packing of three paths or cycles with the minimum number of crossings (three and six, respectively)? Can we compute 1-planar packings with few crossings for triples of other types of trees?
\end{inparaenum}
%
%
%
\bibliography{biblio}

\begin{thebibliography}{10}
\providecommand{\url}[1]{\texttt{#1}}
\providecommand{\urlprefix}{URL }
\providecommand{\doi}[1]{https://doi.org/#1}

\bibitem{DBLP:journals/algorithmica/Brandenburg18}
Brandenburg, F.J.: Recognizing optimal 1-planar graphs in linear time.
  Algorithmica  \textbf{80}(1),  1--28 (2018)

\bibitem{CZAP2012505}
Czap, J., Hudák, D.: 1-planarity of complete multipartite graphs. Discrete
  Applied Mathematics  \textbf{160}(4),  505 -- 512 (2012)

\bibitem{DBLP:journals/csur/DidimoLM19}
Didimo, W., Liotta, G., Montecchiani, F.: A survey on graph drawing beyond
  planarity. {ACM} Comput. Surv.  \textbf{52}(1),  4:1--4:37 (2019)

\bibitem{DBLP:conf/cccg/Frati09}
Frati, F.: Planar packing of diameter-four trees. In: Proceedings of the 21st
  Annual Canadian Conference on Computational Geometry. pp. 95--98 (2009)

\bibitem{DBLP:journals/ipl/FratiGK09}
Frati, F., Geyer, M., Kaufmann, M.: Planar packing of trees and spider trees.
  Inf. Process. Lett.  \textbf{109}(6),  301--307 (2009)

\bibitem{DBLP:journals/jgt/GarciaHHNT02}
{Garc{\'{\i}}a Olaverri}, A., Hernando, M.C., Hurtado, F., Noy, M., Tejel, J.:
  Packing trees into planar graphs. Journal of Graph Theory  \textbf{40}(3),
  172--181 (2002)

\bibitem{DBLP:conf/wads/GeyerHKKT13}
Geyer, M., Hoffmann, M., Kaufmann, M., Kusters, V., T{\'{o}}th, C.D.: Planar
  packing of binary trees. In: {WADS 2013}. LNCS, vol.~8037, pp. 353--364.
  Springer (2013)

\bibitem{DBLP:journals/jocg/GeyerHKKT17}
Geyer, M., Hoffmann, M., Kaufmann, M., Kusters, V., T{\'{o}}th, C.D.: The
  planar tree packing theorem. JoCG  \textbf{8}(2),  109--177 (2017)

\bibitem{hhs-npttk-81}
Hedetniemi, S., Hedetniemi, S., Slater, P.: A note on packing two trees into
  {$K_n$}. Ars Combinatoria  \textbf{11} (01 1981)

\bibitem{KOBOUROV201749}
Kobourov, S.G., Liotta, G., Montecchiani, F.: An annotated bibliography on
  1-planarity. Computer Science Review  \textbf{25},  49 -- 67 (2017)

\bibitem{KORZHIK20081319}
Korzhik, V.P.: Minimal non-1-planar graphs. Discrete Mathematics
  \textbf{308}(7),  1319--1327 (2008)

\bibitem{DBLP:journals/ejc/MaheoSW96}
Mah{\'{e}}o, M., Sacl{\'{e}}, J., Wozniak, M.: Edge-disjoint placement of three
  trees. Eur. J. Comb.  \textbf{17}(6),  543--563 (1996)

\bibitem{oo-tpptt-06}
Oda, Y., Ota, K.: Tight planar packings of two trees. In: 22nd European
  Workshop on Computational Geometry (2006)

\bibitem{DBLP:journals/combinatorica/PachT97}
Pach, J., T{\'{o}}th, G.: Graphs drawn with few crossings per edge.
  Combinatorica  \textbf{17}(3),  427--439 (1997)

\bibitem{Ringel1965107}
Ringel, G.: Ein sechsfarbenproblem auf der kugel. Abhandlungen aus dem
  Mathematischen Seminar der Universität Hamburg  \textbf{29}(1-2),  107--117
  (1965)

\bibitem{DBLP:journals/jct/SauerS78}
Sauer, N., Spencer, J.: Edge disjoint placement of graphs. J. Comb. Theory,
  Ser. {B}  \textbf{25}(3),  295--302 (1978)

\bibitem{WANG1993137}
Wang, H., Sauer, N.: Packing three copies of a tree into a complete graph.
  European Journal of Combinatorics  \textbf{14}(2),  137 -- 142 (1993)

\bibitem{DBLP:journals/gc/WozniakW93}
Wozniak, M., Wojda, A.P.: Triple placement of graphs. Graphs and Combinatorics
  \textbf{9}(1),  85--91 (1993)

\end{thebibliography}
\bibliographystyle{splncs04}

\newpage
\appendix

\section*{Appendix}\label{ap:appendix}

\section{Additional Material for Section~\ref{se:preliminaries}}

\one*
\begin{proof}
	If each $G_i$ is connected, then it has at least $n-1$ edges and therefore any packing of  $G_1,\dots, G_k$  has at least $k(n-1)$ edges; since the complete graph with $n$ vertices has $\frac{n(n-1)}{2}$ edges it must be $k(n-1) \leq \frac{n(n-1)}{2}$, that is $n \geq 2k$. On the other hand, a 1-planar graph has at most $4n-8$ edges, and therefore it must be $k(n-1) \leq 4n-8$, which implies $k \leq 3$. Moreover, if each $G_i$ is connected $\deg_{G_i}(v) \geq 1$ for each $v$ and since $\sum_{i=1}^k \deg_{G_i}(v) \leq n-1$, it must be $\deg_{G_i}(v) \leq n-k$.
\end{proof}

\section{Additional Material for Section~\ref{se:2-paths+cater}}\label{ap:2-paths+cater}
%
%
\begin{figure}[htb]
	\centering
	\subfigure[]{\includegraphics[width=0.25\textwidth, page=1]{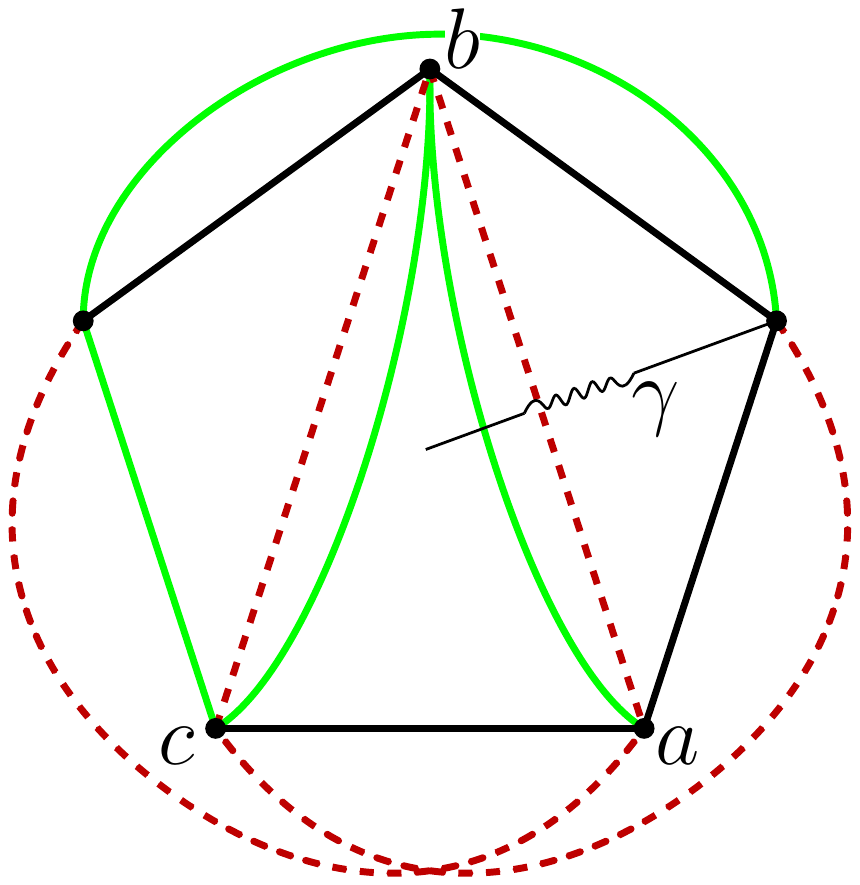}\label{fi:pentagon-a}}
	\hfil
	\subfigure[]{\includegraphics[width=0.25\textwidth, page=2]{pentagon.pdf}\label{fi:pentagon-b}}
	\hfil
	\subfigure[]{\includegraphics[width=0.25\textwidth, page=3]{pentagon.pdf}\label{fi:pentagon-c}}\\
	\hfil
	\subfigure[]{\includegraphics[width=0.25\textwidth, page=2]{5-backbone.pdf}\label{fi:5-backbone-a}}
	\hfil
	\subfigure[]{\includegraphics[width=0.25\textwidth, page=1]{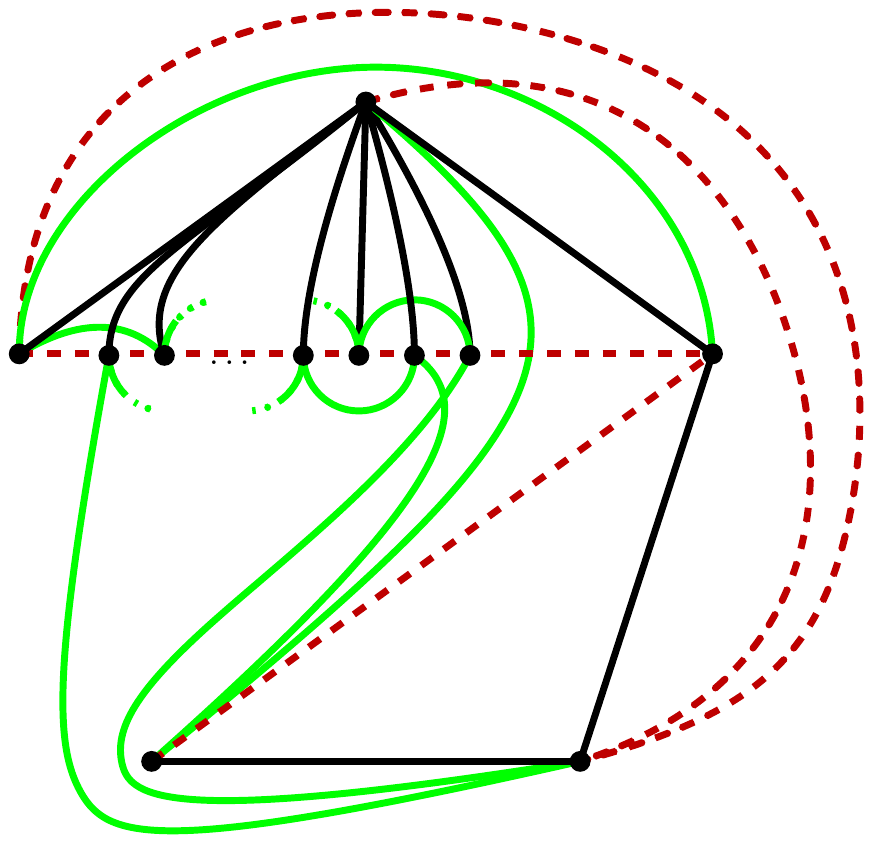}\label{fi:5-backbone-b}}
	\caption{Illustration for the proof of Lemma~\ref{le:5legged-5}.}\label{fi:pentagon}
\end{figure}
\fivelegged*
\begin{proof}
	If $T$ is a path, then $P_1$, $P_2$ and $T$ are all paths of length five, and by Property~\ref{pr:one} a 1-planar packing of $P_1$, $P_2$ and $T$ does not exist. Suppose therefore that at least one internal vertex of the backbone $P$ of $T$ has some leaves attached. In this case we use an approach similar to the one described in the proof of Lemma~\ref{le:5legged-6ormore}. However, as we have just explained, a 1-planar packing of $P'_1$, $P'_2$ and $P$ cannot exist in this case. We start with a 1-planar packing with two pairs of parallel edges. For each pair, one of the two parallel edges belongs to $P'_1$ and the other one to $P'_2$. We will remove the parallel edges by performing the leaf addition operations. To this aim we must guarantee that there is a cutting curve for each pair of parallel edges. The 1-planar packing $P'_1$, $P'_2$ and $P$ and the cutting curves for the internal vertices of $P$ are shown in Fig.~\ref{fi:5-legged-2-c}, for the case when we have at least two vertices with leaves attached. Indeed, if only two vertices have leaves attached, they are either consecutive along the backbone or not. In the first case, the cutting curves of the vertices labeled $a$ and $b$ will remove the parallel edges; in  the second case, the cutting curves of the vertices labeled $a$ and $c$ will remove the parallel edges.
	
	If only one vertex of $P$ has leaves attached, we have only one cutting curve and thus it is not possible to intersect both pairs of parallel edges. To handle this case we distinguish two cases. If the only vertex with leaves attached is the middle vertex of the backbone, then we can adapt the technique used above as follows. Consider the 1-planar packing of $P'_1$, $P'_2$ and $P$ shown in Fig.~\ref{fi:pentagon-a}, where we have two parallel  edges $(a,b)$ and two parallel edges $(b,c)$. Consider now the cutting curve $\gamma$ shown in Fig.~\ref{fi:pentagon-a}. This curve intersects the two parallel edges $(a,b)$, thus, performing a leaf addition operation using that curve, we obtain a 1-planar packing of $P_1$, $P_2$ and $T$ with the two parallel edges $(b,c)$ (see Fig.~\ref{fi:pentagon-b}). These two parallel edges can be removed by modifying the drawing as follows (see also  Fig.~\ref{fi:pentagon-c} for an illustration). Since the two edges crossed by $\gamma$ are parallel edges, the leaf addition operation used must be one of those shown in Fig.~\ref{fi:gadgets-parallel} for parallel edges. No matter which of the cases applies, one of the two edges incident to vertex $a$ is non-crossed and can be disconnected from $a$ and connected to $c$ without introducing any crossing. Call this edge $e$. The parallel edge $(c,b)$ with the same color as $e$ can be disconnected from $c$ and connected to $a$ only crossing $e$. With this modification we obtain the desired 1-planar packing. If the only vertex with leaves attached is the second (or fourth) vertex of the backbone, we compute a 1-planar packing of $P_1$, $P_2$ and $T$ with an ad-hoc technique shown in Figs.~\ref{fi:5-backbone-a} and~\ref{fi:5-backbone-b} for an even or an odd number of leaves, respectively.
\end{proof}

\twopathscaterpillarfourvertices*
\begin{proof}
	Since the length of the backbone is four, we have exactly two non-leaf vertices $v_1$ and $v_2$. Denote by $n_i$ the number of leaves adjacent to $v_i$ ($i=1,2$) and assume $n_1 \leq n_2$. We distinguish different cases depending on the values of $n_1$ and $n_2$. If $n_1=1$, then we have $\deg_T(v_2)=n-1$ and by Property~\ref{pr:one} a 1-planar packing of $P_1$, $P_2$ and $T$ does not exist. Assume now that $n_1 \geq 2$. 
	
	\begin{figure}[htbp]
		\centering
		\subfigure[even-even]{\includegraphics[width=0.32\textwidth, page=1]{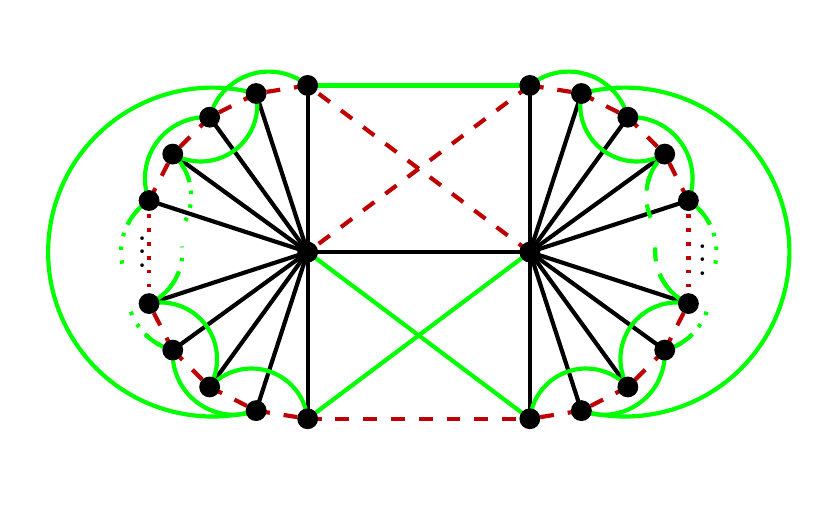}\label{fi:4-backbone-1-a}}
		\hfil
		\subfigure[odd-odd]{\includegraphics[width=0.32\textwidth, page=2]{4-backbone-1.pdf}\label{fi:4-backbone-1-b}}
		\hfil
		\subfigure[even-odd]{\includegraphics[width=0.32\textwidth, page=3]{4-backbone-1.pdf}\label{fi:4-backbone-1-c}}
		\caption{Illustration for the proof of Theorem~\ref{th:caterpillar-4}.}\label{fi:4-backbone-1}
	\end{figure}
	
	We start with the case when $n_1 \geq 5$. In this case we construct a 1-planar packing according to different techniques depending on the parity of $n_1$ and $n_2$. Figs.~\ref{fi:4-backbone-1-a},~\ref{fi:4-backbone-1-b}, and~\ref{fi:4-backbone-1-c} show the construction for the cases when $n_1$ and $n_2$ are both even, when they are both odd, and when they have different parity, respectively. If $n_1 < 5$ we have different ad-hoc constructions that depend on the values of $n_1$ and $n_2$. All cases are shown in Fig.~\ref{fi:4-backbone}. 
\end{proof}

\begin{figure}[htb]
	\centering
	\subfigure[$(2,2)$]{\includegraphics[width=0.15\textwidth, page=1]{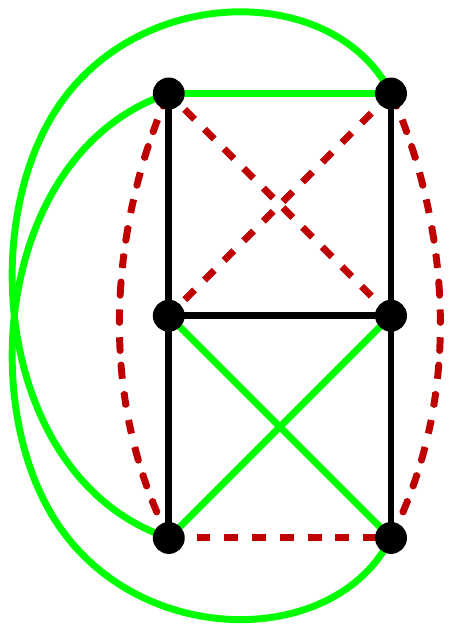}\label{fi:4-backbone-a}}
	\hfil
	\subfigure[$(2,3)$]{\includegraphics[width=0.15\textwidth, page=2]{4-backbone.pdf}\label{fi:4-backbone-b}}
	\hfil
	\subfigure[$(2,4)$]{\includegraphics[width=0.15\textwidth, page=3]{4-backbone.pdf}\label{fi:4-backbone-c}}
	\hfil
	\subfigure[$(2,5^+)$]{\includegraphics[width=0.15\textwidth, page=4]{4-backbone.pdf}\label{fi:4-backbone-d}}
	\hfil
	\subfigure[$(2,6^+)$]{\includegraphics[width=0.15\textwidth, page=5]{4-backbone.pdf}\label{fi:4-backbone-e}}
	\hfil
	\subfigure[$(3,3)$]{\includegraphics[width=0.15\textwidth, page=6]{4-backbone.pdf}\label{fi:4-backbone-f}}
	\hfil
	\subfigure[$(3,4)$]{\includegraphics[width=0.15\textwidth, page=7]{4-backbone.pdf}\label{fi:4-backbone-g}}
	\hfil
	\subfigure[$(3,5^+)$]{\includegraphics[width=0.15\textwidth, page=8]{4-backbone.pdf}\label{fi:4-backbone-h}}
	\hfil
	\subfigure[$(3,6^+)$]{\includegraphics[width=0.15\textwidth, page=9]{4-backbone.pdf}\label{fi:4-backbone-i}}
	\hfil
	\subfigure[$(4,4)$]{\includegraphics[width=0.15\textwidth, page=10]{4-backbone.pdf}\label{fi:4-backbone-j}}
	\hfil
	\subfigure[$(4,5^+)$]{\includegraphics[width=0.15\textwidth, page=11]{4-backbone.pdf}\label{fi:4-backbone-k}}
	\hfil
	\subfigure[$(4,6^+)$]{\includegraphics[width=0.15\textwidth, page=12]{4-backbone.pdf}\label{fi:4-backbone-l}}
	\caption{Illustration for the proof of Theorem~\ref{th:caterpillar-4}. Constructions for the cases when $n_1<5$. For each case the values $(n_1,n_2)$ are indicated; $5^+$ means $n_2 \geq 5$ with $n_2$ odd, while $6^+$ means $n_2 \geq 6$ with $n_2$ even.}\label{fi:4-backbone}
\end{figure}

\section{Additional Material for Section~\ref{se:3-paths}}

\threecycles*
\begin{proof}
	Suppose first that $n \equiv 2 \pmod 3$. In this case, we partition the set $V$ of the $n$ vertices in two groups $V_1$ and $V_2$ of size $7+3k$ and $7+3h$, with $h,k \geq 1$. For each group $V_i$ we compute a 1-planar packing $G_i$ ($i=1,2$) as described in the proof of Theorem~\ref{th:3-paths}. Each $G_i$ has $6$ crossings and it is embedded so that each path has both end-vertices on the external face (each of the three vertices of the external face is the end-vertex of two distinct paths). We create a 1-planar packing of three cycles with $n$ vertices by connecting the two end-vertices of each path in $G_1$ with the two end-vertices of a path in $G_2$. This requires the addition of six edges that can be embedded so to form two crossings (see Fig.~\ref{fi:cycles-a}). Thus, the total number of crossings in the final 1-planar packing is $14$. If $n \equiv 0 \pmod 3$ or $n \equiv 1 \pmod 3$, we proceed in a way similar to the previous case. We create two 1-planar packings $G_1$ and $G_2$ with $7+3k$ and $7+3h$ vertices ($h,k \geq 1$) leaving out one or two vertices. When $G_1$ and $G_2$ are connected to create the 1-planar packing of three cycles we also add the missing one or two vertices as shown in Figs.~\ref{fi:cycles-b} and~\ref{fi:cycles-c}. Also in this case when connecting $G_1$ and $G_2$ we have two additional crossings and a total of $14$ crossings in the final 1-planar packing. 
\end{proof}

\begin{figure}[htb]
	\centering
	\subfigure[]{\includegraphics[width=0.32\textwidth, page=1]{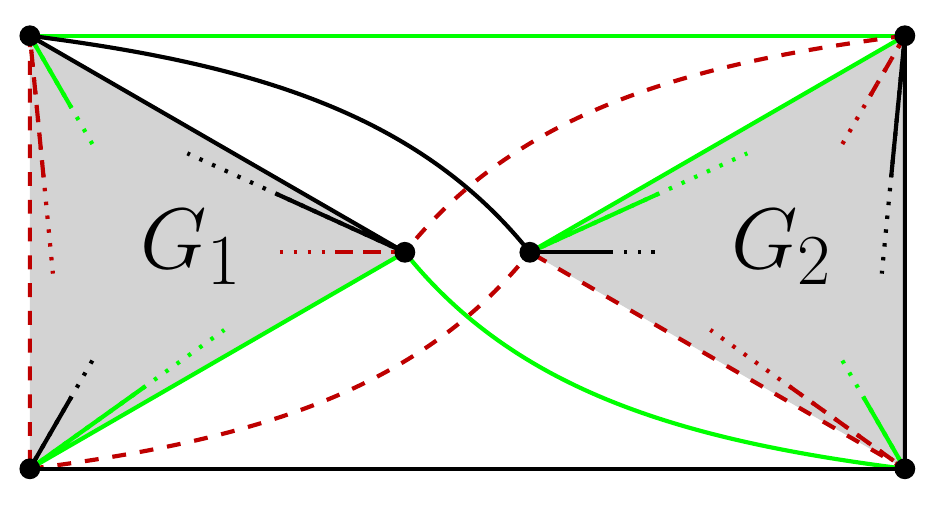}\label{fi:cycles-a}}
	\hfill
	\subfigure[]{\includegraphics[width=0.32\textwidth, page=2]{cycles.pdf}\label{fi:cycles-b}}
	\hfill
	\subfigure[]{\includegraphics[width=0.32\textwidth, page=3]{cycles.pdf}\label{fi:cycles-c}}
	\caption{Illustration for the proof of Theorem~\ref{th:3-cycles}.}\label{fi:cycles}
\end{figure}

%

\newpage 
\section{Additional Material for Section~\ref{se:quadruples}}\label{ap:quadruples}

\begin{figure}[htb]
	\centering
	\subfigure[$n=12$]{\includegraphics[width=0.32\textwidth, page=7]{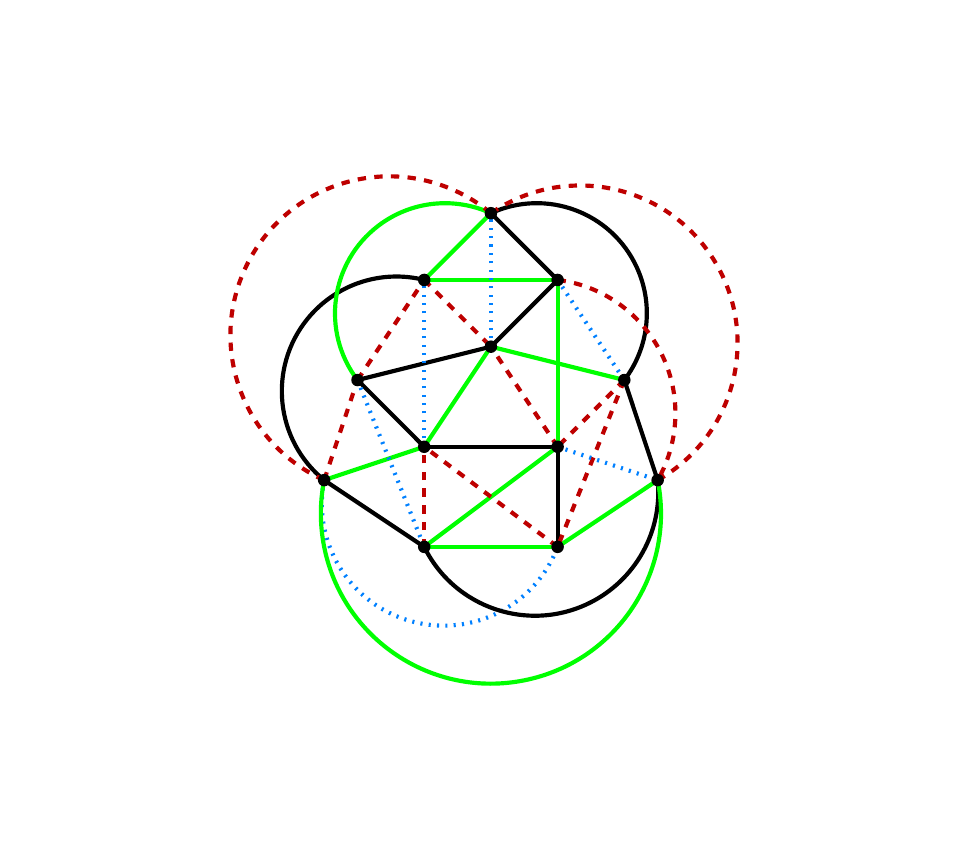}\label{fi:paths-matching-a}}
	\hfill
	\subfigure[$n=14$]{\includegraphics[width=0.32\textwidth, page=8]{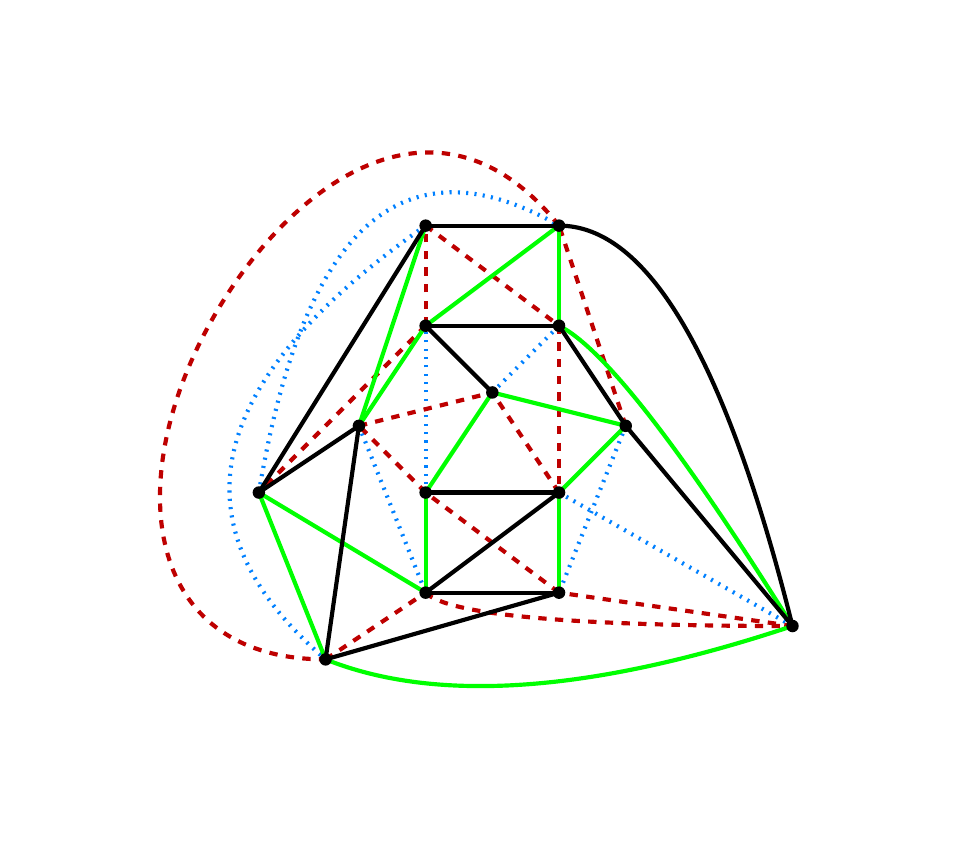}\label{fi:paths-matching-b}}
	\hfill
	\subfigure[$n=16$]{\includegraphics[width=0.32\textwidth, page=9]{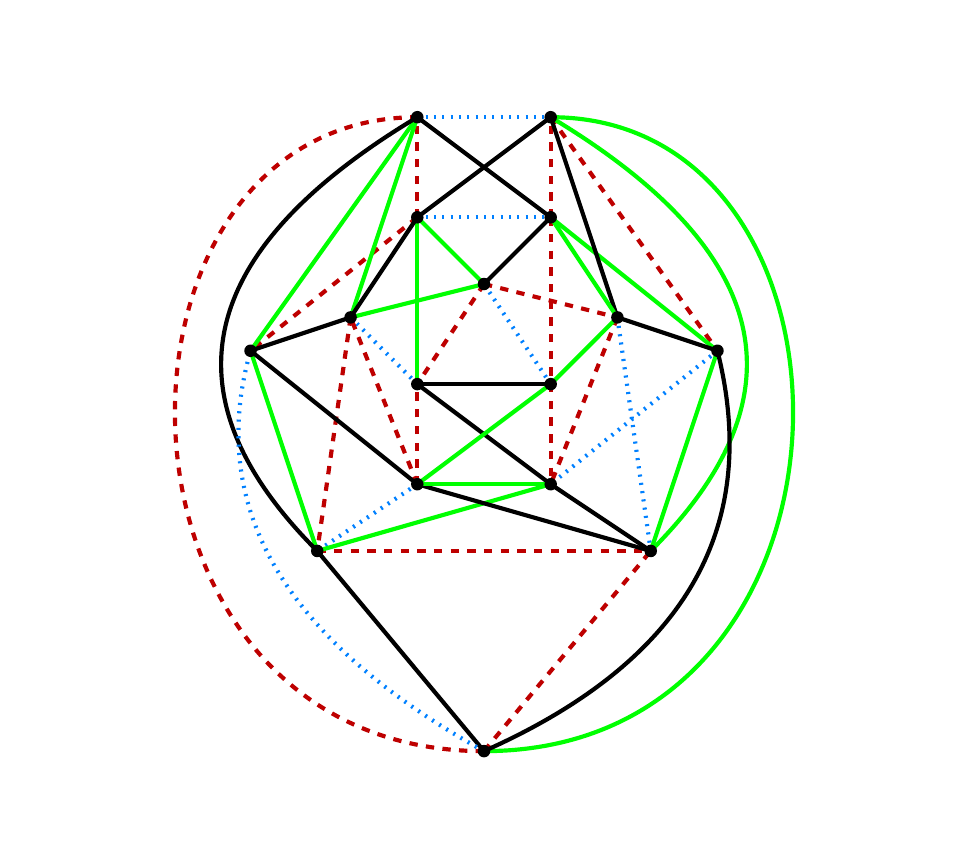}\label{fi:paths-matching-c}}
	\hfill
	\subfigure[$n=18$]{\includegraphics[width=0.32\textwidth, page=10]{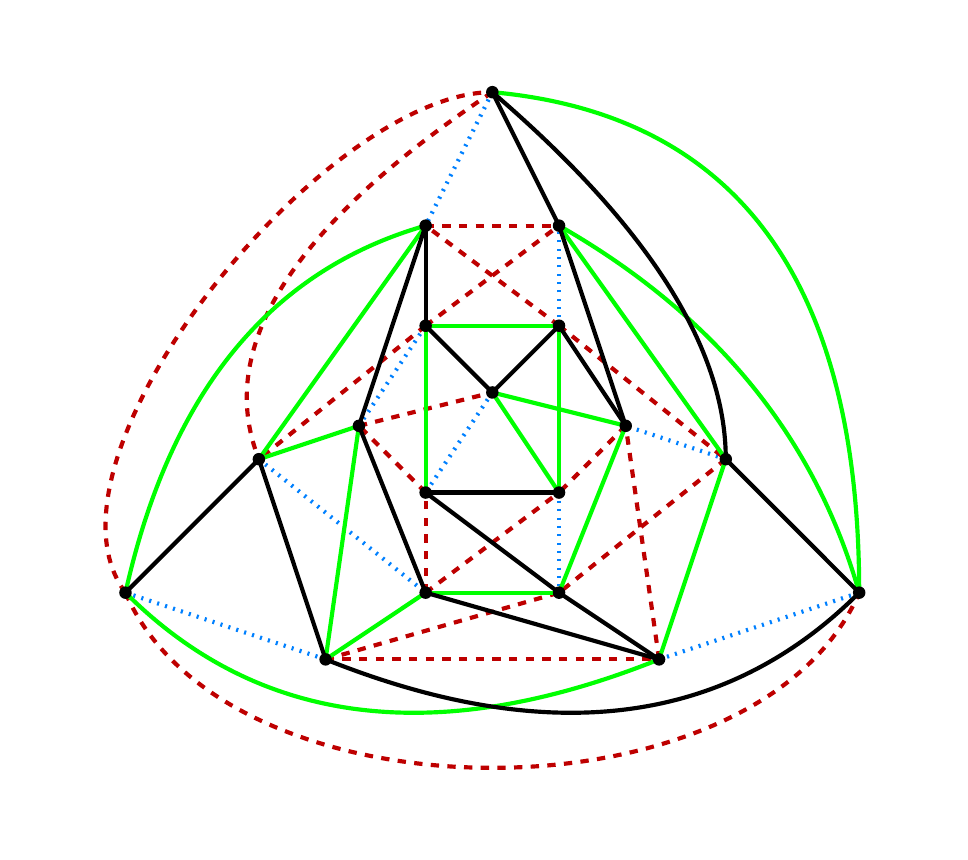}\label{fi:paths-matching-d}}
	\hfill
	\subfigure[$n=20$]{\includegraphics[width=0.32\textwidth, page=11]{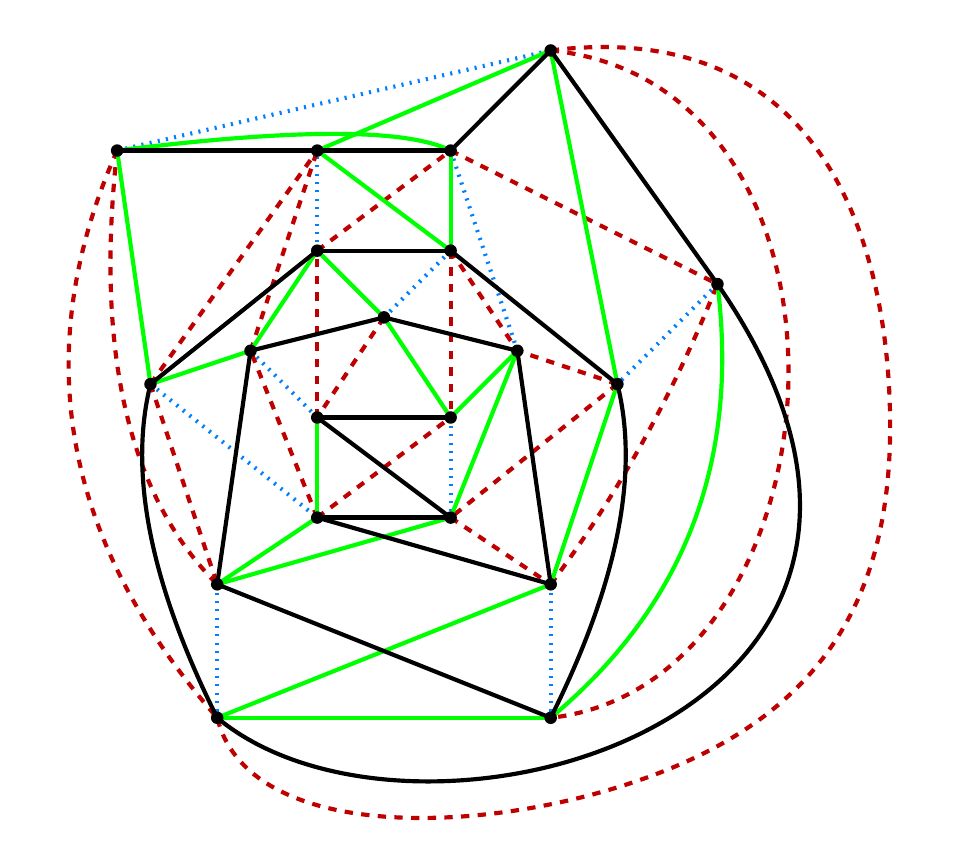}\label{fi:paths-matching-e}}
	\hfill
	\subfigure[$n=22$]{\includegraphics[width=0.32\textwidth, page=12]{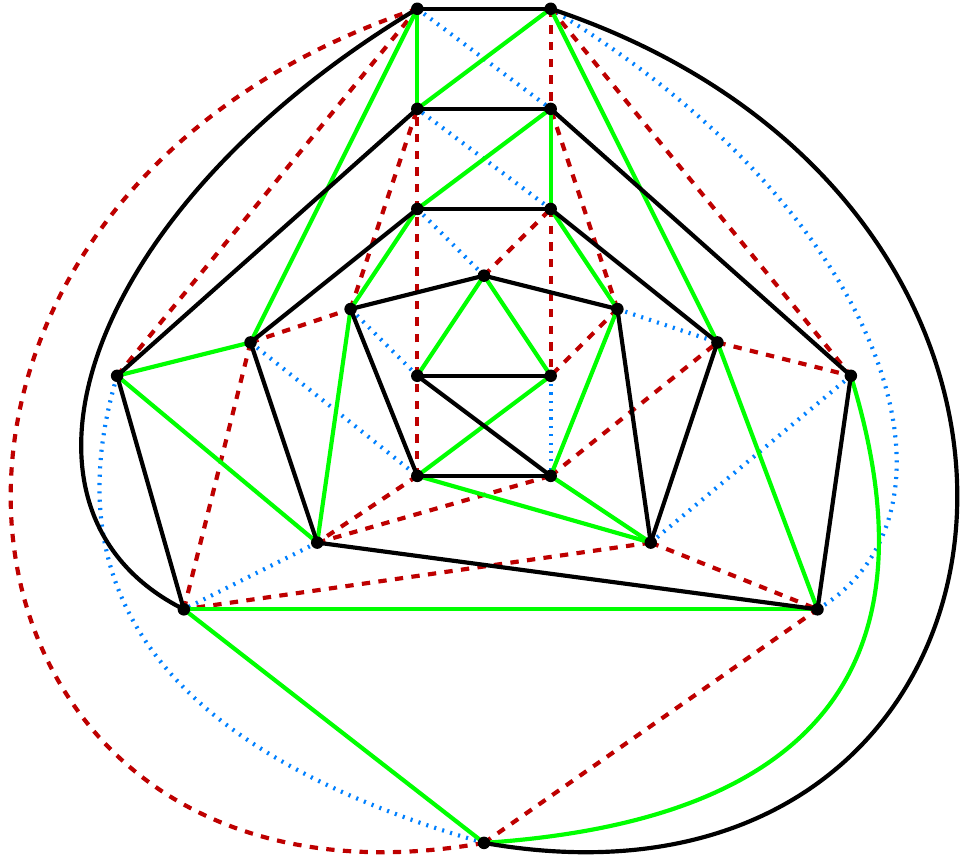}\label{fi:cycles-f}}
	\caption{1-planar packing of three paths and a perfect matching.}\label{fi:paths-matching}
\end{figure}

\end{document}